\documentclass[11pt,letterpaper]{article} 
\usepackage{amsmath}
\usepackage{amssymb}
\usepackage{bbm}

    \usepackage[margin=1in]{geometry}
    \usepackage[utf8]{inputenc}
    \usepackage{amsthm}
    \usepackage{thm-restate}
    \newtheorem{theorem}{Theorem}
    \newtheorem{lemma}[theorem]{Lemma}
    
    \newtheorem{proposition}[theorem]{Proposition}

    \usepackage[pdfpagelabels,bookmarks=false]{hyperref}
    \hypersetup{colorlinks, linkcolor=darkblue, citecolor=black, urlcolor=darkblue}
    \usepackage[capitalize, nameinlink]{cleveref}
    
    \crefname{theorem}{Theorem}{Theorems}
    \Crefname{lemma}{Lemma}{Lemmas}
    \Crefname{proposition}{Proposition}{Propositions}
    \Crefname{claim}{Claim}{Claims}
    \Crefname{fact}{Fact}{Facts}
    \Crefname{remark}{Remark}{Remarks}
    \Crefname{observation}{Observation}{Observations}
    \Crefname{line}{Line}{Lines}
    \crefalias{AlgoLine}{line}
    \Crefname{algocf}{Algorithm}{Algorithms}
    
    \usepackage{times}
    \usepackage[table]{xcolor}

\usepackage{url}
\usepackage{array}
\usepackage{setspace}

\usepackage{lmodern}
\usepackage{mathtools}
\usepackage{bm}
\usepackage{xspace}

\usepackage{xfrac}

\usepackage{tikz}
\usetikzlibrary{calc}
\usetikzlibrary{shapes.geometric}
\usetikzlibrary{arrows}
\usetikzlibrary{decorations.pathreplacing, decorations.pathmorphing}

\usepackage{float}
\usepackage{graphicx}

\usepackage{algpseudocodex}
\usepackage{algorithm}

\newcommand*\smallcircled[1]{\tikz[baseline=(char.base)]{
            \node[shape=circle,draw,inner sep=1pt, minimum size=5pt] (char) {#1};}}

\definecolor{darkblue}{rgb}{0,0,0.38}
\definecolor{darkred}{rgb}{0.8,0,0}
\definecolor{darkgreen}{rgb}{0.1,0.35,0}

\renewcommand{\sfrac}[2]{{\textstyle\frac{#1}{#2}}}

\renewcommand{\epsilon}{\varepsilon}

\newcommand{\LP}{\ensuremath{\mathrm{LP}}}

\newcommand{\ind}{\ensuremath{\mathrm{ind}}}

\def\cupp{\stackrel{.}{\cup}}

\newcommand{\Fscr}{\mathcal{F}}

\newcommand{\Lscr}{\mathcal{L}}

\usepackage{todonotes}

\newcommand{\slack}[1]{\ensuremath{\mathrm{slack}\bigl( #1 \bigr)}}
\newcommand{\cost}[1]{c\bigl( #1 \bigr)}
\newcommand{\lb}[1]{\ell\bigl( #1 \bigr)}
\newcommand{\lbonly}{\ell}
\newcommand{\potPhi}[1]{\Phi\bigl( #1 \bigr)}

\usepackage{natbib}
\renewcommand{\cite}[1]{\citeauthor{#1} [\citeyear{#1}]}

\newcommand{\bib}[3]{\bibitem[\protect\citeauthoryear{#1}{#2}]{#3}}

\begin{document}

\title{Better approximation guarantee for Asymmetric TSP}
\author{Jens Vygen}
\date{\small Research Institute for Discrete Mathematics and Hausdorff Center for Mathematics, University of Bonn}
\maketitle

\begin{abstract}
We improve the approximation ratio for the \textsc{Asymmetric TSP} to less than 15.
We also obtain improved ratios for the special case of unweighted digraphs
and the generalization where we ask for a minimum-cost tour with given (distinct) endpoints.
Moreover, we prove better upper bounds on the integrality ratios of the natural LP relaxations.
\end{abstract}

\section{Introduction}

The \textsc{Asymmetric Traveling Salesman Problem} (ATSP) is one of the best-known combinatorial optimization problems,
with obvious applications in vehicle routing and beyond. 
Given a strongly connected digraph $G=(V,E)$ with non-negative edge costs $c:E\to\mathbb{R}_{\ge 0}$,
we ask for a \emph{tour} (i.e., a multi-edge set $F$ such that $(V,F)$ is connected and Eulerian) of minimum total cost.

Less than ten years ago, \cite{SveTV20} found the first constant-factor approximation algorithm. 
Their approximation ratio of 506 was improved to $22+\epsilon$ by \cite{TraV22} and to $17+\epsilon$ by \cite{TraV25}
(for any $\epsilon>0$). 
Here we improve the approximation ratio again. 
Like the previous works, we also compare to the standard linear programming relaxation:
\begin{equation}
\label{eq:atsplp}
\begin{array}{@{}lrcll}
\multicolumn{2}{@{}l}{\min \ \sum_{e\in E}\, c(e) x_e} & & & \\[1mm]
\mbox{subject to} & x(\delta(U)) &\ge& 2 & (\emptyset\not=U\subset V) \\
& x(\delta^+(v)) &=& x(\delta^-(v)) & (v\in V) \\
& x_e & \ge & 0 & (e\in E),
\end{array}
\end{equation}
whose integral solutions correspond to tours.
Hence we also improve the upper bound on the integrality ratio of this linear program. 
More precisely, we prove:

\begin{theorem}\label{thm:mainatsp}
Let $\alpha$ be a constant with $\alpha > 3 + 2\sqrt{2}$.
Then there is a polynomial-time algorithm that computes, 
for any given strongly connected digraph $G=(V,E)$ with edge costs $c:E\to\mathbb{R}_{\ge 0}$,
a tour with cost at most $(9+\alpha)\cdot \LP$, where $\LP$ denotes the value of \eqref{eq:atsplp}.
If $c(e)=1$ for all $e\in E$, then the tour has at most $\alpha\cdot\LP$ edges.
\end{theorem}

\cite{TraV25} proved this for $\alpha>8$.
Our improvement comes from a relatively small enhancement of Svensson's algorithm, which was originally
developed for the unit-weight special case (which we call \textsc{Asymmetric Graph TSP}) by \cite{Sve15}
and adapted to general ATSP by \cite{SveTV20}.
In each iteration, this algorithm computes a subtour cover, 
which is used either to compute a better initialization that allows us to restart,
or to extend the current partial solution. 
We introduce a new opportunity for a better initialization, which allows us to
prove a better upper bound on the total cost of the edges used to extend the partial solution.

Our result shows that we have still not completely understood the full potential of the algorithm.
We guess that further improvements might be possible.

\section{Revisiting Vertebrate Pairs and Svensson's Algorithm}

This section reviews the setting, the notation, and Svensson's algorithm (without the change that we will make),
following \cite{TraV25} closely.

\subsection{Reduction to Vertebrate Pairs}\label{section:vertebratepairs}

\paragraph{Strongly laminar instances.}
An ATSP instance consists of a strongly connected digraph $G=(V,E)$ with edge costs $c:E\to\mathbb{R}_{\ge 0}$.
\cite{SveTV20} and \cite{TraV22} showed that, without loss of generality, we can assume a certain structure.
We are given a laminar family $\Lscr$ of subsets of $V$ with weights $y\in\mathbb{R}_{\ge 0}^{\Lscr}$.
For every $e\in E$, we have $c(e)=\sum_{L\in\Lscr:\, e\in\delta(L)} y_L$. 
Define $\Lscr_{\ge 2}:=\{L\in\Lscr: |L|\ge 2\}$.
We are also given an optimum LP solution~$x$ and have $x(\delta(L))=2$ for all $L\in \Lscr$, 
and the induced subgraph $G[L]$ is strongly connected for each $L\in\Lscr$.
We write $y_v:=y_{\{v\}}$ for $v\in V$ with $\{v\}\in\Lscr$ and $y_v:=0$ for other vertices $v$.

\paragraph{Backbone.}
A \emph{backbone} $B=(V(B),E(B))$ is a connected Eulerian multi-subgraph of $(V,E)$
such that $V(B)\cap L\not=\emptyset$ for all $L\in \Lscr_{\ge 2}$.
The backbone is fixed throughout this paper.

\paragraph{ATSP via vertebrate pairs.}
\cite{TraV22}, improving on \cite{SveTV20}, showed:

\begin{theorem}[\cite{TraV22}] \label{thm:reductiontovertebratepairs}
Let $\eta\ge 0$. 
Suppose there is a polynomial-time algorithm that computes, for any given strongly laminar ATSP instance 
and any given backbone $B$, 
a multi-edge set $H$ such that $E(B)\cupp H$ is a tour and $c(H)\le 2\cdot\LP + \eta\cdot \sum_{v\in V\setminus V(B)} 2y_v$.
Then there is a polynomial-time algorithm that computes a tour of cost at most $(8+\eta)\cdot\LP$ for any given ATSP instance.
\end{theorem}

In this work, we will reduce $\eta$ to less than $7$.

\paragraph{Asymmetric Graph TSP.}
In \textsc{Asymmetric Graph TSP}, we have $\Lscr=\bigl\{\{v\}:v\in V\bigr\}$ and $y_v=\frac{1}{2}$ for all $v\in V$.
In this case, we work with the empty backbone $B=(\emptyset,\emptyset)$ and do not need 
Theorem~\ref{thm:reductiontovertebratepairs}.
It is important to note that here we will \emph{not} require $x(\delta(L))=2$ for the singletons $L\in\Lscr$.
This allows us to deal with \textsc{Asymmetric Graph TSP} and general \textsc{ATSP} simultaneously.

\paragraph{Main result.}
In this paper we show:

\begin{theorem}\label{thm:main}
Let $\alpha$ be a constant with $\alpha > 3 + 2\sqrt{2}$.
Then there is a polynomial-time algorithm that computes, for any given strongly laminar ATSP instance 
and any given backbone $B$, 
a multi-edge set $H$ such that $E(B)\cupp H$ is a tour and 
$c(H)\le 2\cdot\LP + ( \alpha +1 ) \cdot \sum_{v\in V\setminus V(B)} 2y_v$.
If $\Lscr_{\ge 2}=\emptyset$ and $B=(\emptyset,\emptyset)$, then $c(H)\le \alpha\cdot \LP$.
\end{theorem}

Together with Theorem~\ref{thm:reductiontovertebratepairs}, this directly implies Theorem~\ref{thm:mainatsp}.

\subsection{High-Level Description of Svensson's Algorithm}

We first give a brief high-level outline and then define the necessary terms. 
Svensson's algorithm tries to extend a given backbone to a tour.
In addition to a strongly laminar instance and a backbone~$B$, 
Svensson's algorithm is given a multi-edge set $\tilde H$, called \emph{initialization} 
(see Section~\ref{section:svenssondetails}). 
Initially, $\tilde H=\emptyset$.
After setting $H:=\tilde H$, the algorithm iterates the following three steps
until a \emph{significantly better initialization} was found or $E(B)\cupp H$ is a tour.
\begin{enumerate}
\item[\protect\smallcircled{1}] Compute a \emph{subtour cover} $F$ for $H$.
\item[\protect\smallcircled{2}] Using the edges of $\tilde H$ and $F$, try to find a \emph{significantly better initialization}.
\item[\protect\smallcircled{3}] Add a certain Eulerian multi-edge set to $H$ and check whether $(V,E(B)\cupp H)$ is a tour.
\end{enumerate}
When a significantly better initialization is found in Step \smallcircled{2}, we simply restart the algorithm with this initialization.
When $(V,E(B)\cupp H)$ is a tour after \smallcircled{3}, we stop and output $H$.

In this paper, we will change Step \smallcircled{2} and the analysis in order to obtain a better guarantee.

\subsection{Details of Svensson's Algorithm}\label{section:svenssondetails}

\paragraph{Local.}
An edge $e$ is called \emph{local} if $e\notin\delta(L)$ for all $L\in \Lscr_{\ge 2}$.
A local edge $e=(v,w)$ has cost $c(e)=y_v+y_w$.  
A graph or its edge set is called local if all its edges are local.
In \textsc{Asymmetric Graph TSP}, all edges are local, and the backbone is empty.
For general ATSP, Svensson's algorithm uses only local edges until the very end when it connects everything to the backbone.

\paragraph{Vertex budgets.}
We define budgets $\lbonly : V\setminus V(B) \rightarrow  \mathbb{R}_{>0}$ to pay for local edges. Let $n:=|V|$.
Fix constants $\epsilon$ and $\gamma$ with $0<\epsilon \le \frac{1}{4}$ and $\gamma \ge \frac{2}{1-\epsilon}$.
For \textsc{Asymmetric Graph TSP}, we set
\begin{equation*} 
   \lbonly(v) \ := \ \gamma \cdot x(\delta^-(v))
\end{equation*} 
for all $v\in V$. Here $\lbonly(V) = \gamma\cdot \LP \le \gamma n^2 \le n^2 \cdot \min\{\lbonly(v):v\in V\}$.
For general \textsc{ATSP}, we set
\begin{equation*} 
   \lbonly(v) \ := \ \max\Biggl\{ \gamma \cdot 2y_v, \ \sfrac{\epsilon}{n} \cdot\! \sum_{w\in V\setminus V(B)}2y_w \Biggr\}
\end{equation*}
for all $v\in V\setminus V(B)$. 
Here, 
$\lb{V\setminus V(B)}\le (\gamma+\epsilon)\sum_{v\in V\setminus V(B)}2y_v 
\le \frac{(\gamma+\epsilon)n}{\epsilon} \cdot \min\bigl\{\lbonly(v):v\in V\setminus V(B)\bigr\}$.
So in both cases, we have
\begin{equation}\label{eq:budgetspread}
\lb{V\setminus V(B)} \ \le \ \sfrac{(\gamma+\epsilon)n^2}{\epsilon} \cdot \min\bigl\{ \lbonly(v) : v\in V\setminus V(B) \bigr\}
\end{equation}
and
\begin{equation}\label{eq:budgetlowerbound}
\lbonly(v) \ \ge \ \gamma \cdot x(\delta^-(v)) \cdot 2y_v  \qquad \text{for all } v\in V\setminus V(B).
\end{equation}

\cite{TraV25} set $\gamma$ slightly larger than $4$, but here we will choose a smaller value.

\paragraph{Slack.} 
For a subset $\tilde V$ of $V\setminus B$ and a multi-subset $\tilde E$ of $E$, we write 
\[
\slack{\tilde V,\tilde E} \ := \ \lb{\tilde V} - c \bigl( \tilde E[\tilde V] \bigr).
\]
When we use this notation, $\tilde V$ will be the vertex set of a connected component of $\bigl(V,\tilde E\bigr)$.
We call a multi-subset $H$ of $E$ \emph{light} if 
$\slack{V(D),E(D)}\ge 0$ for every connected component $D$ of $(V,H)$.

\paragraph{Initialization.} An \emph{initialization} is a light local Eulerian multi-subset of $E[V\setminus V(B)]$.
Set $p:=\bigl\lceil\frac{2}{\epsilon}\bigr\rceil$. 
Given an initialization $\tilde H$ such that the connected components of 
$\bigl( V\setminus V(B),\tilde H \bigr)$ have vertex sets $\tilde W_1,\ldots,\tilde W_k$, we write 
\[ 
\potPhi{\tilde H} \ := \ \sum_{i=1}^k \ \slack{\tilde W_i, \tilde H}^{p}.
\]
Given two initializations $\tilde H_1$ and $\tilde H_2$, 
we say that $\tilde H_2$ is \emph{significantly better} than $\tilde H_1$ if
\begin{equation*}
 \potPhi{\tilde H_2} - \potPhi{\tilde H_1} \ > \ \min\bigl\{ \lbonly(v)^{p} : v\in V\setminus V(B) \bigr\}. 
\end{equation*}

\paragraph{Number of restarts.}

Since 
$0\le\potPhi{\tilde H} \le \lb{V\setminus V(B)}^p 
\le \bigl((\gamma+\epsilon)\epsilon^{-1}n^2\bigr)^{p} \cdot \min\bigl\{ \lbonly(v)^{p} : v\in V\setminus V(B) \bigr\}$ 
by \eqref{eq:budgetspread},
there can be at most $O(n^{2p})$ restarts with a significantly better initialization. Recall that $\epsilon$, $\gamma$, and $p$ are constants.

\paragraph{Subtour cover.}
Given an Eulerian multi-subset $H$ of $E[V\setminus V(B)]$, 
a \emph{subtour cover} for $H$ is an Eulerian multi-subset $F$ of $E$ such that
\begin{itemize}
\item $F\cap \delta(W) \ne \emptyset$ for each vertex set $W$ of a connected component of $(V\setminus V(B),H)$. 
\item Every connected component $D$ of $(V,F)$ with $V(D)\cap V(B)\not=\emptyset$ is local.
\end{itemize}

\paragraph{Index.}
We sort the connected components $\tilde W_1,\ldots,\tilde W_k$ of $(V\setminus V(B),\tilde H)$ such that
$$\slack{\tilde W_1,\tilde H} \ \ge \ \cdots \ \ge \ \slack{\tilde W_k,\tilde H}$$
and let $\tilde W_0:=V(B)$.
Then, for a connected multi-subgraph $D$ of $G$, we define the \emph{index} of $D$ to be
\begin{equation*}
\mathrm{ind}(D) \ := \ \min \bigl\{ j\in\{0,\ldots,k\}: V(D)\cap \tilde W_j\not=\emptyset \bigr\}.
\end{equation*}
Note that $\tilde H$ and the sets $\tilde W_0,\ldots,\tilde W_k$ are fixed during an execution of Svensson's algorithm.

\paragraph{Step \protect\smallcircled{3} of Svensson's Algorithm.}
For a current multi-edge set $H$ and a subtour cover $F$ for $H$, this step works as follows.

\vspace*{1mm}
\noindent\rule{\textwidth}{0.4pt}
\begin{algorithmic}
\State{Set $X\gets\emptyset$.} 
\State{\algorithmicrepeat}
\State{\quad Let $Z$ be the connected component of $(V, E(B) \cupp H \cupp F \cupp X)$ with largest $\ind(Z)$.}
\State{\quad \algorithmicif{} there is a local circuit $C$ with $E(C)\cap\delta(V(Z))\not=\emptyset$ and $c(E(C)) \le \slack{\tilde W_{\ind(Z)},\tilde H}$}
\State {\qquad \algorithmicthen{} add the edges of $C$ to $X$.}
\State{\algorithmicuntil{} no such circuit $C$ exists.}
\State{Add the edges of $(F\cupp X)[V(Z)]$ to $H$.}
\end{algorithmic}
\vspace*{-1mm}
\noindent\rule{\textwidth}{0.4pt}
\vspace*{0mm}

Every iteration of the above repeat-loop reduces the number of connected components of $(V, E(B) \cupp H \cupp F \cupp X)$.
Every global iteration of Svensson's algorithm (in particular Step~\smallcircled{3}) 
reduces the number of connected components of $(V,E(B)\cupp H)$.
Hence the algorithm runs in polynomial time. All considered edge sets (including the backbone) are Eulerian, 
so when $(V,E(B)\cupp H)$ is connected, it is a tour.

\subsection{Subtour Cover Guarantee (Step \protect\smallcircled{1} of Svensson's Algorithm)}\label{section:subtourcoverguarantees}

We implement Step \smallcircled{1} of Svensson's algorithm using the following:

\begin{theorem}[\cite{TraV25}]
\label{thm:subtour_cover}
Given a strongly laminar ATSP instance $(G,\Lscr,x,y)$ with $G=(V,E)$, a (possibly empty) backbone $B$, 
and an Eulerian multi-subset $H$ of $E$,
we can compute in polynomial time a subtour cover $F$ for $H$ such that 
\begin{equation}\label{eq:subtourcover_totalcost}
c(F) \ \le \ 2\cdot\LP+\sum_{v\in V\setminus V(B)}2y_v
\end{equation}
and for every connected component $D$ of $(V,F)$ with $V(D)\cap V(B)=\emptyset$ and every vertex $v\in V(D)$ with $y_v>0$, 
\begin{equation} \label{eq:indegree2subtourcover}
|F\cap\delta^-(v)| \ \le \ 2 \cdot x(\delta^-(v)). 
\end{equation}
\end{theorem}
\cite{TraV25} (Theorem 8.24, Definition 7.15) stated the result with a weaker condition 
instead of \eqref{eq:indegree2subtourcover}, but actually proved \eqref{eq:indegree2subtourcover};
see the proof of Theorem 8.24.
For \textsc{Asymmetric Graph TSP}, \eqref{eq:subtourcover_totalcost} is not needed, and there is a simpler proof
(Theorem 6.6 in \cite{TraV25}).

In the next two sections, we give a self-contained proof of Theorem~\ref{thm:main}, using Theorem~\ref{thm:subtour_cover}.

\section{Checking Circuits for a Better Initialization}

Let us illustrate our main new idea in a simple setting, where $y_v$ is the same for all $v\in V\setminus V(B)$ 
and $x(\delta^-(v))=1$ for all $v\in V$.
See also Figure~\ref{fig:newidea_cheap_F}.

In Step \smallcircled{2} of Svensson's algorithm, we first decompose the subtour cover into cycles.
Consider one of these cycles, say $C$ with $\ind(C)=i\ge 1$. 
If $\gamma\cdot \bigl|V(C)\setminus \tilde W_i\bigr| \gg \bigl|E(C)|$, then the graph
$\bigl(\tilde W_i\cup V(C),\tilde H[\tilde W_i]\cup E(C) \bigr)$ has much larger slack than 
$\bigl(\tilde W_i,\tilde H[\tilde W_i] \bigr)$.
This will lead to a significantly better initialization (see Lemma~\ref{lemma:improved_initialization} below),
and we can restart Svensson's algorithm.
Otherwise, using $|E(C)|=|V(C)|$, we obtain, approximately,
$(\gamma-1) \cdot \bigl|V(C)\setminus \tilde W_i\bigr| \le \bigl|V(C)\cap \tilde W_i\bigr|$.

The next observation is that, due to \eqref{eq:indegree2subtourcover}, any vertex of $\tilde W_i$ 
belongs to at most two of those cycles.
We conclude that the total number of vertices in cycles $C$ with $\ind(C)=i$ in our subtour cover decomposition 
is not much more than $2\bigl(1+\frac{1}{\gamma-1}\bigr)\cdot\bigl|\tilde W_i\bigr|$.

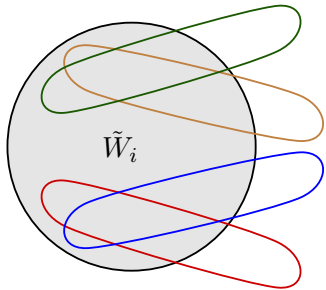
\begin{figure}[hb]
\begin{center}
 \begin{tikzpicture}[scale=1.5, line width=0.7]

\draw[fill=gray, fill opacity =0.2] (3,1) ellipse (1.1 and 1.1) {};
\node at (2.9,1) {$\tilde W_i$};

\draw[darkred] plot [smooth cycle] coordinates {(2.4,0.3) (2.4,0.7) (4.3,0.15) (4.3,-0.25) };
\draw[blue] plot [smooth cycle] coordinates {(2.6,0.1) (2.6,0.5) (4.5,0.95) (4.5,0.55) };
\draw[brown] plot [smooth cycle] coordinates {(2.6,1.9) (2.6,1.5) (4.5,1.05) (4.5,1.45) };
\draw[darkgreen] plot [smooth cycle] coordinates {(2.4,1.7) (2.4,1.3) (4.3,1.85) (4.3,2.25) };

 \end{tikzpicture}
\end{center}
\vspace*{-4mm}
\caption{Assume for simplicity that $y_v$ is the same for all $v\in V$ and $x(\delta^-(v))=1$ for all $v\in \tilde W_i$.
The subtour cover can be decomposed into cycles. 
Among these cycles, consider those that intersect the set $\tilde W_i$ but no connected component of the initialization with smaller index.
Four of these are sketched in the figure in different colors.
If for any of these, the outside part is large compared to the inside part, then one gets a better initialization.
In the simple setting, every vertex of $\tilde W_i$ belongs to at most two of these cycles, 
so the total number of vertices (and hence edges) of the cycles
is not much more than the total number of vertices in $\tilde W_i$.  
\label{fig:newidea_cheap_F}}
\end{figure}

\subsection{New Initialization Lemma}

We now make this idea precise and generalize it to general instances. The first major part is:

\begin{lemma} \label{lemma:newbetterinit}
Let $C$ be a local circuit in $G[V\setminus V(B)]$ with 
$V(C)\setminus \tilde W_{\mathrm{ind}(C)}\not=\emptyset$ and 
\begin{equation}\label{eq:exploringcircuit}
 \cost{E(C)} \ > \ \left(1+\sfrac{1}{\gamma(1-\epsilon)-1}\right) \cdot \sum_{v\in V(C)\cap\tilde W_{\mathrm{ind}(C)} } 2y_v.
\end{equation}
Then we can compute in polynomial time a significantly better initialization.
\end{lemma}

We prove this in two steps. The following is a variant of Lemma 6.9 in \cite{TraV25}:

\begin{restatable}{lemma}{lemmaimprovedinitialization}\label{lemma:improved_initialization}
Let $D$ be a connected Eulerian multi-subgraph of $G[V\setminus V(B)]$ with 
$V(D)\setminus \tilde W_{\mathrm{ind}(D)}\not=\emptyset$ and
\begin{equation}\label{eq:atsp_D_large_and_light}
\slack{V(D),E(D)} \ > \ \slack{\tilde W_{\mathrm{ind}(D)},\tilde H} + \epsilon \cdot \lb{V(D)\setminus \tilde W_{\mathrm{ind}(D)}} .
\end{equation}
Then we can compute in polynomial time a light  Eulerian multi-set $\tilde H '\subseteq \tilde H \cupp E(D)$ such that 
\begin{equation}\label{eq:potentialincrease}
 \potPhi{\tilde H'} - \potPhi{\tilde H} \ > \ \min\bigl\{ \lbonly(v)^{p} : v\in V\setminus V(B) \bigr\}. 
\end{equation}
\end{restatable}

The proof is similar to the proof of Lemma 6.9 in \cite{TraV25}; we defer it to Section~\ref{section:proofofimprovedinit}.
We combine this with the following: 

\begin{lemma} \label{lemma:newbetterslack}
Let $C$ be a local circuit in $G[V\setminus V(B)]$ with 
$V(C)\setminus \tilde W_{\mathrm{ind}(C)}\not=\emptyset$ and \eqref{eq:exploringcircuit}.
Then the digraph $D$ with
$V(D)=\tilde W_{\ind(C)}\cup V(C)$ and $E(D)=\tilde H[\tilde W_{\ind(C)}]\cupp E(C)$ 
satisfies \eqref{eq:atsp_D_large_and_light}.
\end{lemma}

\begin{proof}
We write $i:=\mathrm{ind}(C)$.
First, note that \eqref{eq:exploringcircuit} yields
\begin{align*}
\sum_{v\in V(C)\setminus \tilde W_i} 2y_v &\ = \ \sum_{v\in V(C)} 2y_v - \sum_{v\in V(C)\cap \tilde W_i} 2y_v \\
&\ = \ \cost{E(C)} - \sum_{v\in V(C)\cap \tilde W_i} 2y_v \\
&\ > \ \sfrac{1}{\gamma(1-\epsilon)-1} \cdot \sum_{v\in V(C)\cap\tilde W_i } 2y_v.
\end{align*}
Next (using the above in the strict inequality),
\begin{align*}
(1-\epsilon) \cdot \lb{V(C)\setminus \tilde W_i}
&\ \ge \ \gamma (1-\epsilon) \cdot \sum_{v\in V(C)\setminus \tilde W_i } 2y_v \\
&\ = \ \sum_{v\in V(C)\setminus \tilde W_i } 2y_v + (\gamma(1-\epsilon)-1) \cdot \sum_{v\in V(C)\setminus \tilde W_i } 2y_v \\
&\ > \ \sum_{v\in V(C)\setminus \tilde W_i } 2y_v + \sum_{v\in  V(C)\cap\tilde W_i } 2y_v \\
&\ = \ \cost{E(C)}. 
\end{align*}
This implies:
\begin{align*}
\slack{V(D),E(D)} &\ = \ \slack{\tilde W_i,\tilde H} + \lb{V(C)\setminus \tilde W_i} - \cost{E(C)} \\
&\ > \ \slack{\tilde W_i,\tilde H} + \epsilon \cdot \lb{V(C)\setminus \tilde W_i}.    
\qedhere
\end{align*}
\end{proof}

Lemma~\ref{lemma:newbetterinit} now follows easily:
Given $C$, we construct the connected Eulerian local multi-subgraph $D$ 
as in Lemma~\ref{lemma:newbetterslack}.
By Lemma~\ref{lemma:improved_initialization}, we obtain a significantly better initialization.

\subsection{New Step \protect\smallcircled{2} of Svensson's Algorithm}

Each outer iteration of Svensson's algorithm computes (in Step \smallcircled{1}) 
a subtour cover $F$ for some Eulerian multi-subset $H$ of $E$.
Before proceeding with this, we add the following steps:
\begin{itemize}
\item Decompose $F$ into circuits: write $\chi^F= \sum_{C\in\Fscr} \chi^C$, where $\Fscr$ is a set of circuits in $G$ and
$\chi^F\in\mathbb{Z}_{\ge 0}^{E}$ denotes the incidence vector of $F$: $\chi^F_e$ is the number of copies of $e$ in $F$.
\item Check for each circuit $C\in\Fscr$ whether $C$ is completely contained in a connected component of $(V\setminus V(B),H)$;
if so, remove $C$ from $F$.
Note that this preserves the guarantees from Theorem~\ref{thm:subtour_cover}. 
\item Check for each circuit $C\in\Fscr$ with $\ind(C)>0$ whether $C$ satisfies \eqref{eq:exploringcircuit};
if so, obtain a sig\-nificantly better initialization $\tilde H'$ by Lemma~\ref{lemma:newbetterinit} and restart Svensson's algorithm with 
$\tilde H:=\tilde H'$.
\item Check for each circuit $C\in\Fscr$ with $\ind(C)>0$ whether $D=C$ satisfies \eqref{eq:atsp_D_large_and_light};
if so, obtain a sig\-nificantly better initialization $\tilde H'$ by Lemma~\ref{lemma:improved_initialization} and restart Svensson's algorithm with 
$\tilde H:=\tilde H'$.
\end{itemize}

After these steps,
let $\Fscr_i=\{C\in\Fscr:\mathrm{ind}(C)=i\}$ and $F_i$ the disjoint union of the edge sets of the circuits in $\Fscr_i$
(i.e., $\chi^{F_i}= \sum_{C\in\Fscr_i} \chi^C$).
Note that this definition differs from the previous works: 
there, $F_i$ was defined to be the union of the \emph{connected components} $D$ of $(V,F)$ with $\mathrm{ind}(D)=i$.

\begin{theorem}
Unless the above results in a significantly better initialization, we have
\begin{equation}\label{eq:circuitshavesmallslack}
 \slack{V(C),E(C)} \ \le \ \slack{\tilde W_{\mathrm{ind}(C)},\tilde H} + \epsilon \cdot \lb{V(C)}
\end{equation}
for every circuit $C\in\Fscr$ with $\ind(C)>0$ and
\begin{equation}\label{eq:bound_on_Fi}
 \cost{F_i} \ \le \ \sfrac{2(1-\epsilon)}{\gamma(1-\epsilon)-1} \cdot \lb{\tilde W_i}
\end{equation}
for all $i=1,\ldots,k$.
\end{theorem}

\begin{proof}
Let $C\in\Fscr$ with $\ind(C)>0$. 
Since \eqref{eq:atsp_D_large_and_light} does not hold for $D=C$, we have \eqref{eq:circuitshavesmallslack}.

Moreover, \eqref{eq:exploringcircuit} does not hold for $C$, i.e.,
\begin{equation}\label{eq:nonexploringcircuit}
 \cost{E(C)} \ \le \ \left(1+\sfrac{1}{\gamma(1-\epsilon)-1}\right) \cdot \sum_{v\in V(C)\cap\tilde W_{\mathrm{ind}(C)} } 2y_v.
\end{equation}
for all $C\in\Fscr$ with $\ind(C)>0$. 

Let $i\in\{1,\ldots,k\}$.  
Summing the inequalities \eqref{eq:nonexploringcircuit} for all $C\in\Fscr_i$ yields
\begin{equation*}
 \cost{F_i} \ \le \ \left(1+\sfrac{1}{\gamma(1-\epsilon)-1}\right) \cdot \sum_{C\in\Fscr_i} \sum_{v\in V(C)\cap\tilde W_i } 2y_v
 \ = \ \sfrac{\gamma(1-\epsilon)}{\gamma(1-\epsilon)-1} \cdot \sum_{C\in\Fscr_i} \sum_{v\in V(C)\cap\tilde W_i } 2y_v.
\end{equation*}
On the other hand, \eqref{eq:indegree2subtourcover} yields
\begin{equation*}
\sum_{C\in\Fscr_i} \sum_{v\in V(C)\cap\tilde W_i } 2y_v
\ = \ \sum_{v\in\tilde W_i} \bigl| F_i\cap \delta^-(v) \bigr| \cdot 2y_v
\ \le \ \sum_{v\in\tilde W_i} \bigl| F\cap \delta^-(v) \bigr| \cdot 2y_v
\ \le \ \sum_{v\in\tilde W_i} 2\cdot x(\delta^-(v)) \cdot 2y_v.
\end{equation*}
Together (and using \eqref{eq:budgetlowerbound}) we get
\begin{equation*}
 \cost{F_i} \ \le \ \sfrac{2\gamma(1-\epsilon)}{\gamma(1-\epsilon)-1} \cdot \sum_{v\in\tilde W_i} x(\delta^-(v)) \cdot 2y_v
 \ \le \ \sfrac{2(1-\epsilon)}{\gamma(1-\epsilon)-1} \cdot \lb{\tilde W_i}.
 \qedhere
\end{equation*}
\end{proof}

\subsection{From More Slack to Better Initialization: Proof of Lemma~\ref{lemma:improved_initialization}}\label{section:proofofimprovedinit}

Lemma~\ref{lemma:improved_initialization} is a variant of the main lemma that \cite{TraV25} used to find a better initialization $\tilde H'$
(their Lemma 6.9).
We are able to find such a better initialization whenever we have a connected Eulerian multi-subgraph $D$ of $G[V\setminus V(B)]$ 
that has significantly larger slack than every connected component of $\bigl( V\setminus V(B),\tilde H \bigr)$ it intersects. 
The main difference to Lemma 6.9 in \cite{TraV25} is the term 
$\lb{V(D)\setminus \tilde W_{\mathrm{ind}(D)}}$ instead of $\lb{V(D)}$ on the right-hand side of \eqref{eq:atsp_D_large_and_light}.
On the other hand, the conclusion in \eqref{eq:potentialincrease} is slightly weaker than in \cite{TraV25}, but this is irrelevant.
The construction of $D^*$ in the proof is identical too, but the last part of the proof is different.

\lemmaimprovedinitialization*

\begin{proof}
If $\slack{\tilde W_{\mathrm{ind}(D)},\tilde H} < \min \bigl\{ \lbonly(v) : v\in V\setminus V(B) \bigr\}$,
we simply set $\tilde H' := \tilde H\setminus \tilde H[\tilde W_{\mathrm{ind}(D)}]$, which does the job
because we replace the connected component induced by $\tilde W_{\mathrm{ind}(D)}$ by
at least two singletons, each of which has slack at least $\min\bigl\{ \lbonly(v) : v\in V\setminus V(B) \bigr\}$.
Henceforth, we assume
$\slack{\tilde W_{\mathrm{ind}(D)},\tilde H} \ge \min \bigl\{ \lbonly(v) : v\in V\setminus V(B) \bigr\}$.

Let $I := \bigl\{ j\in\{1,\ldots,k\}: V(D)\cap \tilde W_j\not=\emptyset \bigr\}$ and $i:=\min I =\mathrm{ind}(D)$.
Let
\begin{equation}
J \ := \ \left\{ j\in I \ : \ \lb{\tilde W_j \cap V(D)} \ \le \ \slack{\tilde W_j,\tilde H} \right\}.
\end{equation}
We replace the components $\tilde H[\tilde W_j]$ for $j \in I$
by one new component that is the union of $E(D)$ and all $\tilde H[\tilde W_j]$ with $j\in J$, as well as possibly some singleton components. 
More precisely, we set 
\begin{equation*}
\tilde H ' \ :=\ \bigcup_{h\in \{1,\ldots,k\}\setminus I} \tilde H[\tilde W_h]\ \cupp\ E(D)\ \cupp\ \bigcup_{j\in J} \tilde H[\tilde W_j].
\end{equation*}

Let $D^*$ be the connected component of $(V, \tilde H')$ with edge set
\begin{equation*}
 E(D) \cupp \bigcup_{j\in J} \tilde H[\tilde W_j].
\end{equation*}
Then
\begin{equation*} 
\begin{aligned}
\slack{V(D^*),E(D^*)} &\ = \ \slack{V(D),E(D)} 
+ \sum_{j\in J} \left( \slack{\tilde W_j,\tilde H} - \lb{\tilde W_j \cap V(D)}  \right) \\
&\ \ge \ \slack{V(D),E(D)},
\end{aligned}
\end{equation*}
so $E(D^*)$ is light, and hence $\tilde H '$ is also light.

We will now show \eqref{eq:potentialincrease}. To this end, we will first prove another lower bound on $\slack{V(D^*),E(D^*)}$.
Using the definition of $J$ and \eqref{eq:atsp_D_large_and_light}, we get
\begin{equation} \label{eq:lowerboundslack_newcomponent2}
\begin{aligned} 
& \hspace*{-8mm} 
\slack{V(D^*),E(D^*)} \\[1mm]
& = \ \slack{V(D),E(D)} + \sum_{j\in J} \left( \slack{\tilde W_j,\tilde H} - \lb{\tilde W_j \cap V(D)} \right) \\
& \ge \ \slack{V(D),E(D)} + \sfrac{\epsilon}{2} \sum_{j\in J\setminus\{i\}} \left( \slack{\tilde W_j,\tilde H} - \lb{\tilde W_j \cap V(D)}  \right)  \\
& \ge \ \slack{V(D),E(D)} + \sfrac{\epsilon}{2} \sum_{j\in I\setminus\{i\}} \left( \slack{\tilde W_j,\tilde H} - \lb{\tilde W_j \cap V(D)}  \right) \\
& > \ \slack{\tilde W_i,\tilde H} + \epsilon \cdot \lb{V(D)\setminus\tilde W_i}  
+ \sfrac{\epsilon}{2} \sum_{j\in I\setminus\{i\}} \left( \slack{\tilde W_j,\tilde H} - \lb{\tilde W_j \cap V(D)} \right) \\
& = \ \slack{\tilde W_i,\tilde H} + \sfrac{\epsilon}{2} \cdot \lb{V(D)\setminus\tilde W_i}
+ \sfrac{\epsilon}{2} \sum_{j\in I\setminus\{i\}} \slack{\tilde W_j,\tilde H}.
 \end{aligned}
\end{equation}

Using \eqref{eq:lowerboundslack_newcomponent2}  
and the definition of $p\ge  \sfrac{2}{\epsilon}$, we get
\begin{align*}
&  \hspace*{-8mm} \slack{V(D^*),E(D^*)}^{p} \\
 &\ > \ \Biggl( \slack{\tilde W_i,\tilde H} + \sfrac{\epsilon}{2} \cdot \lb{V(D)\setminus\tilde W_i}
 + \sfrac{\epsilon}{2} \cdot \sum_{j\in I\setminus\{i\}} \slack{\tilde W_j,\tilde H} \Biggr)^{p} \\
&\ \ge \ \slack{\tilde W_i,\tilde H}^{p} 
+ p \cdot \sfrac{\epsilon}{2} \cdot \slack{\tilde W_i,\tilde H}^{p-1} 
\Biggl( \lb{V(D)\setminus\tilde W_i} + \sum_{j\in I\setminus\{i\}} \slack{\tilde W_j,\tilde H} \Biggr) \\
&\ \ge \ \slack{\tilde W_i,\tilde H}^{p} 
+ \slack{\tilde W_i,\tilde H}^{p-1} \Biggl( \lb{V(D)\setminus\tilde W_i} + \sum_{j\in I\setminus\{i\}} \slack{\tilde W_j,\tilde H} \Biggr) \\
&\ \ge \ \slack{\tilde W_i,\tilde H}^{p} + \min\bigl\{\lbonly(v):v\in V\setminus V(B)\bigr\}^p + \sum_{j\in I\setminus\{i\}} \slack{\tilde W_j,\tilde H}^p,
\end{align*}
where we used 
$\slack{\tilde W_i,\tilde H} \ge\min \bigl\{\lbonly(v) : v\in V\setminus V(B) \bigr\}$ and
$V(D)\setminus\tilde W_i\not=\emptyset$ and
$\slack{\tilde W_i,\tilde H} \ge \slack{\tilde W_j,\tilde H}$ for all $j\in I$ in the last inequality.

We conclude
\begin{align*}
 \potPhi{\tilde H'} - \potPhi{\tilde H} &\ \ge \ \slack{V(D^*),E(D^*)}^p - \sum_{j\in I} \slack{\tilde W_j,\tilde H}^p \\
&\ > \ \min\bigl\{\lbonly(v):v\in V\setminus V(B) \bigr\}^p.
\qedhere
 \end{align*}
\end{proof}

\section{Cost of the Tour Computed by Svensson's Algorithm}

If Svensson's algorithm outputs a tour $E(B)\cupp H$ (and not a significantly better initialization), we need to 
bound the cost of $H$.
Note that $H$ has three types of edges: 
\begin{itemize}
\item Those that are part of the initialization $\tilde H$; they cost at most $\lb{V\setminus V(B)}$ because $\tilde H$ is light;
\item Those that we added in Step~\smallcircled{3} as part of $X$; 
these \emph{$X$-edges} cost at most $\lb{V\setminus V(B)}-c\bigl(\tilde H\bigr)$ (cf.\ Proposition~\ref{prop:Xedges} below);
\item Those that we added in Step~\smallcircled{3} as part of $F$; 
the cost of these \emph{$F$-edges} is bounded by our new Lemma~\ref{lemma:Fedges} below.
\end{itemize}

\begin{proposition}[\cite{TraV25}]
\label{prop:Xedges}
The total cost of the $X$-edges added to $H$ during Svensson's algorithm is at most $\lb{V\setminus V(B)}-c\bigl(\tilde H\bigr)$.
\end{proposition}

\begin{proof}
We assign to each set $\tilde W_i$ ($i=1,\ldots,k$) the budget $ \slack{\tilde W_i,\tilde H}$ 
and show that all $X$-edges can be paid from these budgets.
If we add a circuit $C$ to $X$ and later to $H$, then $c(E(C)) \le \slack{\tilde W_{\ind(Z)},\tilde H}$ and we pay $C$ from the
budget of $\tilde W_{\ind(Z)}$. Since $C$ connects
$Z$ to a connected component of $(V,E(B)\cupp H\cupp F\cupp X)$ with smaller index, 
no other circuit will later be paid from the same budget.
\end{proof}

So the edges in $\tilde H$ and the $X$-edges together cost at most $\lb{V\setminus V(B)}$. 
This is exactly as in \cite{TraV25}. Our improvement stems from a better bound on the cost of the $F$-edges:

\begin{lemma}\label{lemma:Fedges}
The total cost of all $F$-edges that are added to $H$ is at most 
\begin{equation*}
\sfrac{2(1-\epsilon)}{\gamma(1-\epsilon)-1} \cdot \lbonly(V)
\end{equation*}
if $B=\emptyset$ and 
\[
2\cdot\LP \, + \sum_{v\in V\setminus V(B)}2y_v \, + \, 
\sfrac{2(1-\epsilon)}{\gamma(1-\epsilon)-1} \cdot \lb{V\setminus V(B)}
\]
otherwise.
\end{lemma}

\begin{proof}
Svensson's algorithm executes Step~\smallcircled{3} in each iteration; let us number these $1,\ldots,t^{\max}$.
Let $Z^t$ denote $Z$ at the end of the $t$-th execution of Step~\smallcircled{3}.
Let $\Fscr^t_i$ and $F^t_i$ denote $\Fscr_i$ and $F_i$ at that time if the edges in $F_i$ are then added to $H$, 
and let $F^t_i=\emptyset$ otherwise.

If $B\not=\emptyset$ and $F^t_0$ is nonempty for some $t$, then this is the very last execution $t=t^{\max}$, and 
we have $\cost{F^{t^{\max}}} \le 2\cdot\LP+\sum_{v\in V\setminus V(B)}2y_v$
by \eqref{eq:subtourcover_totalcost}. 

If $F^t_i$ is nonempty for some $i\ge 1$ and some $t\le t^{\max}$ with $F^t_0=\emptyset$, then 
$\cost{F^t_i} \le \sfrac{2(1-\epsilon)}{\gamma(1-\epsilon)-1} \cdot \lb{\tilde W_i}$ by \eqref{eq:bound_on_Fi}.
We claim that for any $i\ge 1$, at most one of the $F^t_i$ is nonempty.
Then summing over all $i$ and $t$ concludes the proof.

Suppose there are $t_1<t_2$ such that $F^{t_1}_i\not=\emptyset$ and $F^{t_2}_i\not=\emptyset$.
Then $F^{t_1}_i\subseteq E(Z^{t_1})$ and thus $\tilde W_i\subseteq V(Z^{t_1})$.
Moreover, each circuit in $\Fscr^{t_2}_i$ contains a vertex of $\tilde W_i$ and is not completely contained in $Z^{t_1}$  (such circuits would be discarded immediately).
Let $C$ be such a circuit.
See Figure~\ref{fig:astp_f_edges}.

If $|E(C)| \le \slack{\tilde W_{\mathrm{ind}(Z^{t_1})}, \tilde H}$,
this is a contradiction to reaching the end of execution $t_1$ with the component $Z^{t_1}$.
Otherwise, using $\mathrm{ind}(C) = i \ge \mathrm{ind}(Z^{t_1})$,
we get
\begin{align*}
\sum_{v\in V(C)} 2y_v \ = \ \cost{E(C)} &\ > \ \slack{\tilde W_{\mathrm{ind}(Z^{t_1})},\tilde H} \ \ge \ \slack{\tilde W_{\mathrm{ind}(C)},\tilde H}.
\end{align*}

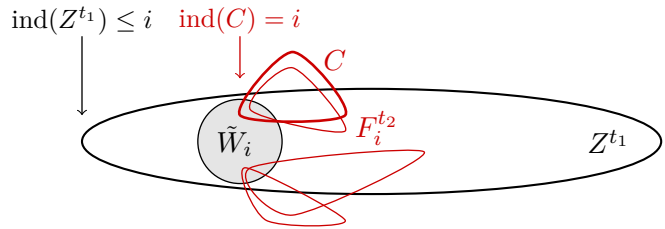
\begin{figure}[bh]
\begin{center}
 \begin{tikzpicture}[scale=0.7, line width=0.5]
  \tikzset{BackwardEdge/.style={->, >=latex, line width=1pt,darkgreen}}
  \tikzset{ForwardEdge/.style={->, >=latex, line width=1pt,darkred}}
  \tikzset{NeutralEdge/.style={->, >=latex, line width=1pt,gray}}

\draw[thick] (5.5,1) ellipse (5.5 and 1) {};
\node at (10,1) {$Z^{t_1}$};
\node at (0,3.3) {\small $\mathrm{ind}(Z^{t_1})\le i$};
\draw[->] (0,3) to (0,1.5);
\draw[fill=gray, fill opacity =0.2] (3,1) ellipse (0.8 and 0.8) {};
\node at (2.9,1) {$\tilde W_i$};

\draw[darkred] plot [smooth cycle] coordinates {(3.1,0.5) (3.5,-0.5) (5,-0.5) };
\draw[darkred] plot [smooth cycle] coordinates {(3.2,0.5) (4,-0.4) (6.5,0.8) };
\draw[darkred, line width=1] plot [smooth cycle] coordinates {(3,1.5) (5,1.5) (4,2.7) };
\draw[darkred] plot [smooth cycle] coordinates {(3.2,1.6) (4,2.4) (5,1.2) };
\node[darkred] at (4.8,2.55) {$C$};
\node[darkred] at (5.6,1.25) {$F^{t_2}_i$};
\node[darkred] at (3,3.3) {\small $\mathrm{ind}(C)=i$};
\draw[darkred,->] (3,3) to (3,2.2);

 \end{tikzpicture}
\end{center}
\vspace*{-4mm}
\caption{Proof of Lemma~\ref{lemma:Fedges}: An example of $F^{t_2}_i$ is shown in red.
Every circuit $C\in\Fscr^{t_2}_i$ contains an edge of $\delta(V(Z^{t_1}))$.
\label{fig:astp_f_edges}}
\end{figure}

Moreover, using $\lbonly(v)\ge\gamma \cdot 2y_v$ for all $v\in V(C)$ and $\gamma \ge \frac{2}{1-\epsilon}$,
we get
\begin{align*}
 \slack{V(C),E(C)} 
 \ &= \ \lb{V(C)} - \cost{E(C)} \\
 \ &\ge \ \epsilon \cdot \ell(V(C)) +  \bigl( (1-\epsilon)\gamma - 1 \bigr) \cdot \sum_{v\in V(C)} 2y_v \\
 \ &\ge \ \sum_{v\in V(C)} 2y_v + \epsilon \cdot \ell(V(C)).
\end{align*}
Together, we obtain
\begin{equation*}
 \slack{V(C),E(C)} \ >\ \slack{\tilde W_{\mathrm{ind}(C)},\tilde H} + \epsilon \cdot \lb{V(C)}.
\end{equation*}
a contradiction to \eqref{eq:circuitshavesmallslack} in execution $t_2$.
\end{proof}

\paragraph{Proof of Theorem~\ref{thm:main}.}
Recalling that the total cost of the edges in the initial $H=\tilde H$ plus the total cost of the $X$-edges added to $H$ 
is at most $\lb{V\setminus V(B)}$, we conclude that 
the returned edge set $H$ has total cost at most 
\begin{equation} \label{eq:boundwithgamma}
\Bigl( 1 + \sfrac{2(1-\epsilon)}{\gamma(1-\epsilon)-1} \Bigr) \cdot \lb{V\setminus V(B)} 
\end{equation}
if $B=\emptyset$, and $2\cdot\LP \, + \sum_{v\in V\setminus V(B)}2y_v$ more otherwise.
Recall that
$\lbonly(V) = \gamma \cdot \LP$ for \textsc{Asymmetric Graph TSP} and 
$\lb{V\setminus V(B)} \le (\gamma+\epsilon) \cdot  \sum_{v\in V\setminus V(B)}2y_v$ for general ATSP. 
Hence it suffices to set $\gamma$ and $\epsilon$ so that
\[
\Bigl( 1 + \sfrac{2(1-\epsilon)}{\gamma(1-\epsilon)-1} \Bigr) \cdot \bigl( \gamma+\epsilon \bigr) \ \le \ \alpha.
\]

Setting $\gamma=1+\sqrt{2}$ and $\epsilon=\min\Bigl\{\frac{1}{10},\,\sfrac{\alpha-3 - 2\sqrt{2}}{6} \Bigr\}$ does the job.
This concludes the proof of Theorem~\ref{thm:main}.

\paragraph{Remark.}
Theorem~\ref{thm:main} says that we have a $(2,4+2\sqrt{2}+\epsilon)$-algorithm vor vertebrate pairs, for any $\epsilon>0$.
Exercises 7.5--7.7 of \cite{TraV25} lead to an additional improvement of approximately 0.15.
The overall guarantee for \textsc{ATSP} is then less than 14.7.

\section{Final Remarks}

\subsection{Can the Bound Be Tight?}

In Step \smallcircled{3} of Svensson's algorithm, we check whether there is a local circuit $C$ with 
               \begin{itemize}
                \item $E(C)\cap\delta(V(Z)) \not=\emptyset$ and
                \item $\cost{E(C)} \le \slack{\tilde W_{\mathrm{ind}(Z)},\tilde H}$,
               \end{itemize}
          then add $E(C)$ to $X$. 
We could change the last condition to $\cost{E(C)} \le \sfrac{1}{\sqrt{2}-3\epsilon}\,\slack{\tilde W_{\mathrm{ind}(Z)},\tilde H}$,
and the proof of Lemma~\ref{lemma:Fedges} would still work.
Hence the bound is tight only if $\cost{\tilde H} \approx \lb{V\setminus V(B)}$, 
and $\slack{\tilde W_i,\tilde H}$ is (on average) much smaller than $\lb{\tilde W_i}$.

If a cycle $C\in\Fscr_i$ or a connected component of $(V,F_i)$ 
has much larger slack than $\slack{\tilde W_i,\tilde H}$,
then this leads to a better initialization.
Hence $\Fscr_i$ contains many cycles and many connected components in the worst case.
For example, pairs of cycles form a connected component.
However, we do not know how to exploit this.

\subsection{Path Version and Other Problems}

\paragraph{Path version.}
Theorem 9.23 of \cite{TraV25} says that a $(\kappa,\eta)$-algorithm for vertebrate pairs implies
an approximation guarantee of $\max\{3\kappa+\eta+2,\, 4\kappa+7\}$ (with respect to the natural LP relaxation) for the 
\textsc{Asymmetric Path TSP}.
We have $\kappa=2$ and $\eta<7$ by Theorem~\ref{thm:main}.
Hence we get an approximation ratio of 15, and an upper bound of 15 for the integrality ratio of the natural LP relaxation
of \textsc{Asymmetric Path TSP}. (The previously best bound from \cite{TraV25} was 17.)

For the unweighted special case, \textsc{Asymmetric Graph Path TSP}, 
we get a $(2\alpha-1)$ -approximation algorithm from any $\alpha$-approximation algorithm for ATSP (\cite{KohTV20}).
Hence we improve the approximation ratio for \textsc{Asymmetric Graph Path TSP} 
(and the upper bound on the integrality ratio of the natural LP relaxation)
from $15+\epsilon$ to less than $11$.

\paragraph{Other problems.}
We also get improved approximation guarantees for all problems where the currently best guarantee is obtained with
a black-box reduction to an ATSP approximation algorithm. 
Examples include asymmetric capacitated vehicle routing and the prize-collecting ATSP (cf.\ Chapter~17 of \cite{TraV25}).

\end{document}